\documentclass{article}

\usepackage[letterpaper]{geometry}
\usepackage[hyphens]{url}
\usepackage[colorlinks=true,citecolor=blue,linkcolor=blue]{hyperref}
\pdfoutput=1 
\usepackage[utf8x]{inputenc}
\usepackage{amsmath,amssymb,enumerate,amsthm}
\usepackage{stackengine}
\usepackage{tikz}
\usepackage[only,shortleftarrow,shortrightarrow]{stmaryrd}
\usetikzlibrary{arrows,calc,positioning,backgrounds}
\usetikzlibrary{decorations.pathmorphing}
\tikzset{>=stealth',auto,node distance=3cm,text height=1ex,
  every label/.append style={font=\scriptsize, label distance=-0.1cm}}

\newtheorem{theorem}{Theorem}[section]
\newtheorem{lemma}{Lemma}[section]
\newtheorem{conjecture}{Conjecture}[section]
\theoremstyle{definition}

\theoremstyle{remark}

\renewcommand{\epsilon}{\varepsilon}
\newcommand{\eps}{\varepsilon}

\usepackage{listings}

\title{Progress on a perimeter surveillance problem}

\author{Jeremy Avigad and Floris van Doorn}

\begin{document}

\maketitle

\begin{abstract}
  We consider a perimeter surveillance problem introduced by Kingston, Beard, and Holt in 2008
  and studied by Davis, Humphrey, and Kingston in 2019.
  In this problem, $n$ drones surveil a finite interval,
  moving at uniform speed and exchanging information only when they meet another drone.
  Kingston et al.\ described a particular online algorithm for coordinating their behavior
  and asked for an upper bound on how long it can take before the drones are fully synchronized.
  They divided the algorithm's behavior into two phases,
  and presented upper bounds on the length of each phase based on conjectured worst-case configurations.
  Davis et al.\ presented counterexamples to the conjecture for phase 1.

  We present sharp upper bounds on phase 2
  which show that in this case the conjectured worst case is correct.
  We also present new lower bounds on phase 1 and the total time to synchronization,
  and report partial progress towards obtaining an upper bound.
\end{abstract}


\section{Introduction}
\label{section:introduction}

In 2008, Kingston, Beard, and Holt \cite{kingston:beard:holt:08} considered
a problem in decentralized control in which a group of
small unmanned aerial vehicles (UAVs) or drones is required to
surveil a linear interval with changing borders.
In their model, the drones all move along the interval
at the same uniform speed and can exchange information only when they meet.
Because the borders of the interval and the number of operant drones
can change over time, the drones have imperfect information as to the
global state.

Kingston et al.\ described an algorithm
for coordinating the drones
and considered the problem of bounding the time to synchronization
in a setting where the parameters remain fixed.
They divided the algorithm's behavior into two phases
which we will call \emph{phase 1} and \emph{phase 2}.
Normalizing units so that a single drone can traverse the interval in one unit of time,
they claimed an upper bound of 3 units of time for phase 1
and an upper bound of 2 units of time for phase 2.
In each case, the bounds were based on the behavior of what they took to be
the worst-case starting configurations.

In 2019, Davis, Humphrey, and Kingston \cite{davis:et:al:19} pointed out
that the previous work did not justify the claimed characterizations of the worst-case behavior,
and, in fact, they provided a counterexample to the bound for phase 1.
As a result, there is currently no rigorously established bound on the time to synchronization
that is independent of the number of drones.

We describe the problem more precisely in Section~\ref{section:the:problem} and
discuss the previous results in greater detail.
In Section~\ref{section:upper:bounds}, we present sharp upper bounds
on the length of phase 2,
showing that the originally claimed worst-case behavior is correct.
In Section~\ref{section:lower:bounds}, we present improved lower bounds on the length of phase 1 as well as the total time to synchronization.
Finally, in Section~\ref{section:algebraic:calculation}, we present some partial results towards bounding the length of phase 1.

It is by now well understood that decentralized coordination of UAVs raises interesting combinatorial challenges \cite{bakule:08,burgard:et:al:05,sujit:beard:08}. What the Kingston--Beard--Holt example shows is that difficult combinatorial problems arise even when dealing with fairly simple models, and that new mathematical ideas and techniques are needed to handle them.

David Greve at Collins Aerospace has independently obtained
a substantially different proof of our Theorem~\ref{theorem:upper:bounds}
with $2$ replacing $2 - 1/n$ \cite{greve:unp},
and he has recently formalized the proof presented here (personal communication)
in the ACL2 verification system \cite{boyer:moore:98}.

\section{The problem}
\label{section:the:problem}

To describe the model under consideration,
it is convenient to normalize units so that the drones are surveilling the unit interval $[0, 1]$
and moving at a velocity of one unit per unit time.
At each point in time, each drone has a direction $d = \pm 1$,
where $1$ indicates that the drone is moving to the right
and $-1$ indicates that it is moving to the left.
Each drone also has an estimate of the form $((a, \ell), (b, m))$
where $a$ is the left endpoint of the interval, $\ell$ is the number of drones to the left,
$b$ is the right endpoint, and $m$ is number of drones to the right.
Kingston et al.\ wanted to consider a scenario where the data keeps changing,
so these estimates may be wrong; in particular, $a$ and $b$ do not need to be in the unit interval.
Each drone recognizes the leftmost or rightmost border when it reaches it.
Two drones can only exchange information when they meet, that is, occupy the same position in the interval.

Suppose we have $n$ drones on the unit interval, numbered from left to right $1, \ldots, n$.
The $i$th drone's \emph{left endpoint} is $(i-1)/n$, its \emph{right endpoint} is $i/n$,
and its \emph{interval} is the closed interval with those endpoints.
We say the \emph{common endpoint} of drones $i$ and $i + 1$ is the right endpoint of
drone $i$,
which is equal to the left endpoint of drone $i + 1$.
The desired situation is that each drone remains in its interval,
moving back and forth between its left and right endpoints.

Kingston et al.\ proposed the following algorithm to attain this behavior.
Write $(\alpha, \beta) = ((a, \ell), (b, m))$ for the drone's estimates.
Based on this data,
each drone can calculate the interval it \emph{thinks} it is supposed to surveil as follows:
\begin{itemize}
\item The size of the interval is $I(\alpha, \beta) = (b - a) / (\ell + m + 1)$, that is,
the length of the estimated interval divided by the number of drones.
\item The left endpoint is
\begin{align*}
  L(\alpha, \beta) &= a + \ell I(\alpha, \beta) = b - (m + 1) I(\alpha, \beta)\\
  &= \frac{a (m + 1) + b \ell}{\ell + m + 1}.
\end{align*}
\item The right endpoint is
\begin{align*}
  R(\alpha, \beta) &= a + (\ell + 1) I(\alpha, \beta) = b - m I(\alpha, \beta) \\
  &= \frac{a m + b (\ell+1)}{\ell + m + 1}.
\end{align*}
\end{itemize}
According to the algorithm, each drone continues in the direction it is moving
until one of these events occurs:
\begin{itemize}
  \item If a drone hits the left border,
  it updates its left estimate with the correct left endpoint of the interval
  and the fact that there are no drones to the left, and then it turns around.
  The case where a drone hits the right border is handled similarly.
  \item If two drones meet,
  the left one adopts the right estimate from the drone to its right
  (adding 1 to the number of drones to the right), and vice-versa.
  As a result, the two drones agree as to their estimates of the intervals they are supposed to surveil.
  The two then set their directions so that they are headed to their common endpoint.
  \item If two drones are traveling together (with consistent estimates) and reach their common endpoint,
  they split, i.e.~one of them reverses direction in order to stay in its estimated interval.
\end{itemize}
We assume that when two or more drones start together,
they all share their estimates at time $0$.
We call the first type of event a \emph{border} event, the second type of event a \emph{meet} event,
and the third type of event a \emph{separation} event.
Notice that, as a special case, the second and third can happen simultaneously,
if two drones meet at their common endpoint. We call that a \emph{bounce} event. Fig.~\ref{figure:sample:run} depicts a sample run of the algorithm with five drones,
with the interval extending left to right and time flowing downward.

\begin{figure}
  \begin{center}
  \includegraphics[width=0.45\textwidth]{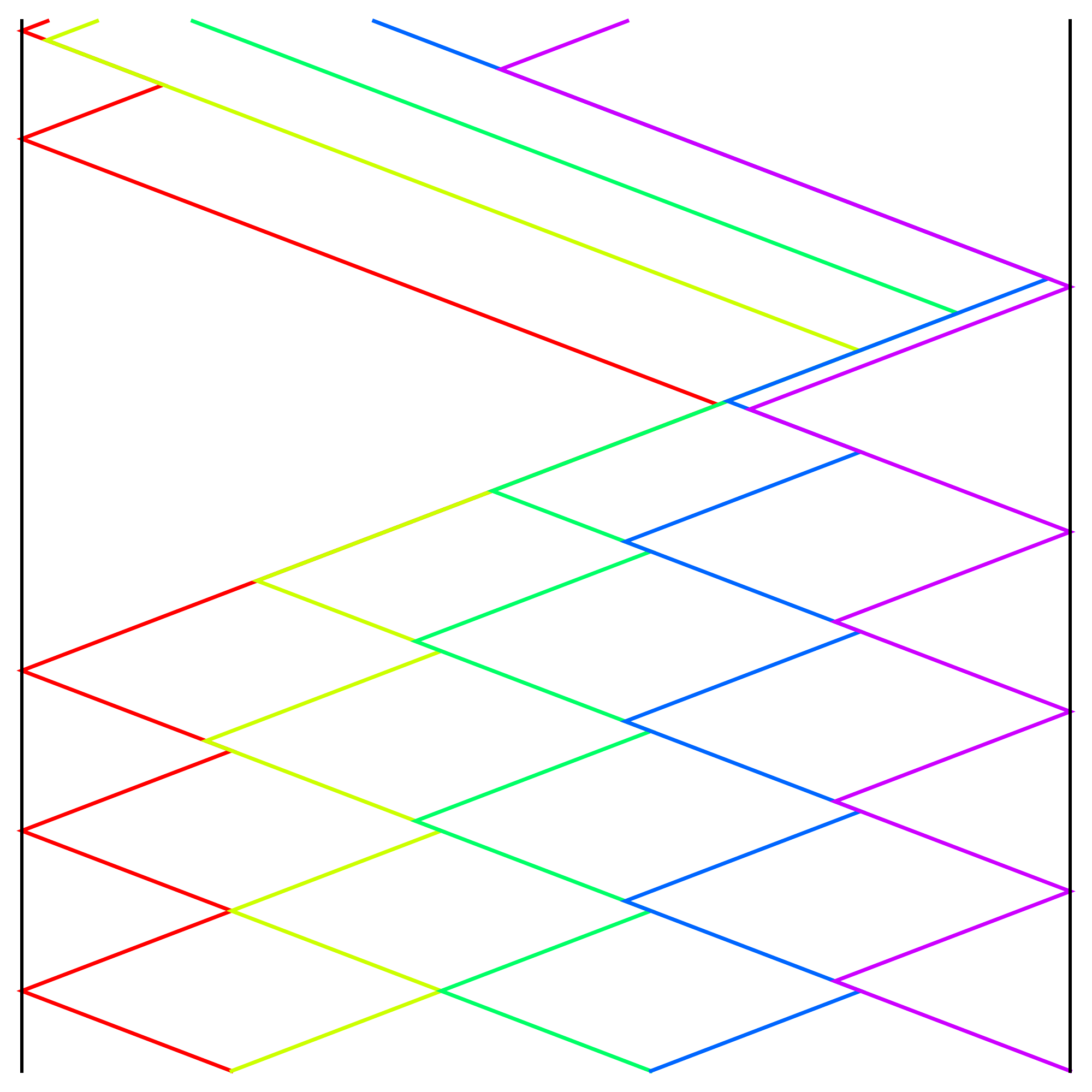}
  \end{center}
  \caption{A sample run of the algorithm with $n = 5$ drones. }
  \label{figure:sample:run}
\end{figure}

The arguments below can be made rigorous by formalizing the notion of a \emph{configuration}
(that is, a time, $t$, and the sequence of positions, directions and estimates of all the drones at that time),
and, given a configuration, the next configuration at which one of the events above occurs.
It then makes sense to talk about the sequence of eventful configurations
from a given start configuration,
and all the informal claims below can be interpreted in terms of that.

One should also establish that the algorithm does not exhibit \emph{Zeno} behavior, i.e.~that for every time $t$ there is a finite sequence of events that extends past time $t$. To prove this, suppose otherwise, consider the infinite sequence of events determined by the algorithm, and let $T$ be the least upper bound on their times. Then for every $\eps > 0$, there are infinitely many events in the time interval $(T-\eps, T)$. Since there are only finitely many drones, at least one drone has to be involved in infinitely many events, which means that one drone has to change direction infinitely many times. But it is not hard to show that if drone $i+1$ makes two consecutive left turns, then drone $i$ must turn right in the interim; so if drone $i + 1$ changes direction infinitely often, then so does drone $i$. With the symmetric argument, it therefore follows that \emph{all} the drones change direction infinitely often in the interval $(T-\eps, T)$. But now if we take $\eps < 1 / 2n$, then at time $T - \eps$ either there is a pair of drones $i$ and $i + 1$ that are not within $2 \eps$ of each other, or the leftmost drone is not within $\eps$ of the left border, or the rightmost drone is not within $\eps$ of the right border. In the first case, drone $i$ can turn at most once in $(T-\eps, T)$, and in the last two cases, the relevant drone can turn at most once. Thus we have contradicted the assumption that the algorithm exhibits Zeno behavior.

Given a particular start configuration, say a drone is \emph{left synchronized} at time $t$
if beyond that point it never goes to the left of its left endpoint,
and similarly for \emph{right synchronized}.
A drone is \emph{synchronized} at time $t$ if it 
is left and right synchronized.
Kingston et al.\ conjectured the following:

\begin{conjecture}
\label{conjecture:a}
From any start configuration, all drones have correct estimates by time 3.
\end{conjecture}

\begin{conjecture}
\label{conjecture:b}
If all drones have correct estimates at time $t$, then all drones are synchronized by time $t + 2$.
\end{conjecture}

\noindent We call the time between the start and the moment that all drones have correct estimates \emph{phase 1}, and the time after phase 1 until the moment that the drones are synchronized \emph{phase 2}.

Kingston et al.\ sketched a proof of each conjecture,
in each case based on a claims that a certain start configurations gave rise to the worst-case behavior.
The two conjectures imply that for any start configuration, the drones are synchronized by time 5.

Davis et al.\ showed that Conjecture~\ref{conjecture:a} is false,
by exhibiting counterexamples with $n = 3$ that require up to $3 + 1/2$ units of time before all the drones have correct estimates.
For that purpose, they used the AGREE model checker \cite{cofer:et:al:12}, which required fixed bounds on all the parameters.
In particular, they had to limit the estimates of the number of drones to the left or right at 20.
With those restrictions, the tool reported upper bounds of $3 + 2/3$ on the time until all three drones have complete information, and $4 + 1/3$ units of time until full synchronization.
The tool also reported absolute upper bounds of 2 on phase 2, with $n \le 6$; they report that the verification for $n = 6$ required about 20 days of computation using 40 cores.
They do not report any results for larger $n$. In particular, there was no rigorously established bound on the length of either phase, or total time to synchronization, that is independent of $n$.

Conjecture~\ref{conjecture:b} is clearly implied by the following statement:
if all drones start with correct estimates then they are synchronized by time 2.
The implication follows, because we can consider the configuration at time $t$ as the new start configuration.
In Section~\ref{section:upper:bounds}, we prove the following:
\begin{theorem}
\label{theorem:upper:bounds}
Assuming all the drones have the correct estimates, they are all synchronized at time $2 - 1/n$.
\end{theorem}
\noindent This shows that the conjecture by Kingston et al.\
as to the worst-case configurations is correct.
In Section~\ref{section:upper:bounds}, we also obtain the following additional information:
\begin{theorem}
  \label{theorem:upper:bounds:combined}
  If all drones start with incorrect estimates, and they all have correct estimates at time $t$,
  then all drones are synchronized by time $t + 1 - 1/n$.
\end{theorem}

In Section~\ref{section:lower:bounds}, we improve the lower bounds as follows:

\begin{theorem}
\label{theorem:lower:bounds}
For every $n\ge 3$ and $\varepsilon > 0$, there is a start configuration
such that drones do not have correct estimates before time $4 - 1 / n - \varepsilon$,
and are not fully synchronized before time $5 - 3 / n - \varepsilon$.
\end{theorem}

We mention in passing that the case $n = 1$ is trivial; a single drone is already synchronized, though it may not have correct estimates until time 2.
The case $n = 2$ is also easy to analyze; drones have correct estimates by time 2 and are synchronized by time $2 + 1/2$.
Both these bounds are sharp, which can be seen by having both drones start together
near the left border of the interval, moving right, and having them separate near the right border of the interval. So $n = 3$ is the first interesting case.

It is important to note that Kingston et al.\
were not looking for an algorithm to synchronize all the drones as quickly as possible.
For that, having all drones move all the way to the left to get the correct information about the left border
and then move all the way to the right does better than the one proposed.
Rather, they were independently interested in the behavior of that particular algorithm
for updating information in the face of changing borders
and addition or subtraction of drones.
Given that, the question about worst-case behavior even under fixed conditions is natural.

The description of the algorithm leaves two things unspecified.
First, it does not specify whether the information that each drone has
must be consistent with its current position. For example, it does not specify
whether a drone can think that the right border is at position 0.8
when the drone itself it is at position 0.9.
Second, it does not specify what happens when
a group of three or more drones come together
and determine that three of them are within the middle drone's interval;
in that case, the middle drone can escort either neighbor to their common
border.
Neither of these issues bears on the results reported below, since
our upper bound only concerns phase 2,
where these issues do not arise,
and our lower bounds meet the more stringent requirement
that all drones have information consistent with their positions.
Regarding the first issue, we note in passing that one can show
that if the drones start with consistent information,
temporary inconsistencies do not affect the behavior of the algorithm.
Regarding the second issue, we note that the strongest upper bound
will allow for nondeterminism and allow the middle drone
to go to either endpoint.

\section{An upper bound on phase 2}
\label{section:upper:bounds}

Once the drones have the correct estimates as to the left and right endpoints
and the number of drones on either side,
each drone knows its proper interval,
and the behavior of the algorithm from that point on can be described more simply:
when two drones meet, they escort each other to their common endpoint and then separate.
Our goal is to prove Theorem~\ref{theorem:upper:bounds},
which guarantees that the drones are all synchronized within $2 - 1/n$ units of time.

By symmetry, it suffices to show that all the drones are left synchronized
by time $2 - 1/n$.
Kingston, Beard, and Holt \cite{kingston:beard:holt:08} gave a short argument that all drones are
eventually left synchronized, although the bound that is implicit in that
argument is linear in $n$.
The argument goes as follows: suppose at some time, $t$, drones $1, \ldots, j$
are left synchronized.
Eventually, drone $j$ will meet drone $j + 1$, and then they will travel to their
common endpoint and separate.
It suffices to show that at this point, $j + 1$ in left synchronized,
because then by induction we have that all drones are eventually left synchronized.
We present their proof of this in Lemma~\ref{lemma:left:synchronized:after:separation}.
Our proof of Theorem~\ref{theorem:upper:bounds} is based on a subtle refinement of their argument.

\begin{lemma}
\label{lemma:direction:change}
Suppose that at time $t$ drone $j$ is moving to the right.
Then the next time drone $j$ changes direction,
it is at or to the right of its right endpoint.
The same is true with ``right'' replaced by ``left.''
\end{lemma}

\begin{proof}
If drone $j < n$ is moving to the right, the only two events in which it can change direction
is when meeting drone $j+1$ or separating from drone $j+1$.
If the next time drone $j$ turns left is when meeting drone $j + 1$,
then at that point they are to the right of their common endpoint, which is the right endpoint for drone $j$.
If the next time drone $j$ turns left is when separating from (or bouncing off) drone $j+1$,
then at that point they must both be at their common endpoint.

If drone $n$ is moving to the right, the only event in which it changes direction is a border event,
which is at its right endpoint.
The last observation follows by symmetry.
\end{proof}

\begin{figure}
  \begin{center}
  \begin{tikzpicture}[scale=2.4]
    \path (0,0) edge[dashed] (0,-1.2);
    \path (0.5,0) edge[thick, blue] (0,-0.5);
    \path (0,-0.5) edge[thick, blue, ->] (0.6,-1.1);
    \path (0.5,0) edge[dashed] (0.5,-1.2);
    \path (0.5,0) edge[thick, red] (1.1,-0.6);
    \path (1.1,-0.6) edge[thick, red, ->] (0.6,-1.1);
    \path (1,0) edge[dashed] (1,-1.2);
  \end{tikzpicture}
  \end{center}
  \caption{A picture proof of Lemma~\ref{lemma:left:synchronized:after:separation}. The two intervals between the vertical lines indicate the intervals of drone $j$ (blue) and drone $j+1$ (red).}
  \label{figure:left:synchronized}
\end{figure}
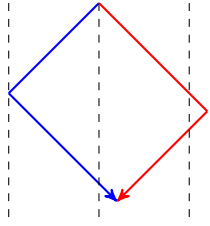

\begin{lemma}
\label{lemma:left:synchronized:after:separation}
Suppose at time $t$, drone $j$ is left synchronized
and drone $j$ and $j+1$ separate at their common endpoint.
Then drone $j+1$ is left synchronized at time $t$.
\end{lemma}

\begin{proof}
By induction, we show that every subsequent meet, bounce, and separation event
involving $j$ and $j + 1$ occurs to the right of their common endpoint.

After the separation, drone $j$ is moving to the left.
By Lemma~\ref{lemma:direction:change}, the next time it changes direction,
it is at or to the left of its left endpoint.
Since it is left synchronized,
we know it is \emph{at} its left endpoint.
Similarly, after the separation, drone $j+1$ is moving to the right,
and the next time it changes direction is it at or to the right of its right endpoint.
So the next time drone $j$ and $j+1$ meet,
they are at or to the right of their common endpoint,
since drone $j+1$ must has taken at least as long to turn around as drone $j$;
see Fig.~\ref{figure:left:synchronized} for a visual depiction.
They then travel left to their common endpoint and separate, and the situation repeats.
\end{proof}

We now draw out two useful consequences of Lemma~\ref{lemma:left:synchronized:after:separation}:

\begin{lemma}
\label{lemma:left:synchronized:sufficient:a}
Suppose at time $t$ drone $j$ and $j+1$ are together moving to the left,
and drone $j$ is left synchronized.
Then drone $j + 1$ is also left synchronized at time $t$.
\end{lemma}

\begin{proof}
If they are together and moving to the left, they are to the right of their common endpoint.
Eventually they will reach their common endpoint and separate, say, at time $t'$.
By Lemma~\ref{lemma:left:synchronized:after:separation},
drone $j + 1$ is left synchronized at time $t'$.
But since drone $j + 1$ is to the right of its left endpoint between time $t$ and $t'$,
it is in fact left synchronized at time $t$.
\end{proof}

\begin{lemma}
\label{lemma:left:synchronized:sufficient:b}
Suppose at time $t$ drone $j < n$ is at or to the right of its right endpoint, moving right,
and left synchronized.
Then drone $j + 1$ is left synchronized at time $t$.
\end{lemma}

\begin{proof}
Since drone $j$ is moving right and is to the right to its right endpoint, it cannot be together with drone $j+1$.
Eventually drone $j$ and $j+1$ will meet to the right of their common endpoint,
say at time $t'$,
and then they will move left together.
At that point, by Lemma~\ref{lemma:left:synchronized:sufficient:a},
drone $j + 1$ is left synchronized.
Since drone $j + 1$ is to the right of its left endpoint between time $t$ and $t'$,
it is already left synchronized at time $t$.
\end{proof}

We have now arrived at the key refinement of the argument by Kingston et al.\.
We say that drones $j$ and $j + 1$ \emph{have met} by time $t$
if either they started together, moving in the same direction,
or they have been involved in a meet or bounce event.
It turns out that we have much more information about the behavior of the drones
once this is the case.
Fortunately, it is not hard to show that this happens within one unit of time,
for all the drones uniformly.

\begin{lemma}
For every $j < n$, drones $j$ and $j+1$ have met by time 1.
\end{lemma}

\begin{proof}
Intuitively, the worst case is where $j$ and $j+1$ start close together
with $j$ moving to the left and $j+1$ moving to the right.
Eventually, $j$ turns around at or before it reaches $0$
and $j+1$ turns around at or before it reaches $1$, and then $j$ and $j + 1$ will meet.
At that point, together they have traveled at most the twice the length of the interval,
which means that each one has traveled at most one unit of distance.

We can make this argument more rigorous as follows.
Suppose drone $j$ starts at position $x$ and drone $j + 1$ starts at position $y \ge x$.
Furthermore, let $w$ be the position of drone $j$ when it first moves right
(so $w = x$ if $j$ starts moving right,
and otherwise $w$ is the position where drone $j$ first turns around),
and let $z$ be the position of drone $j+1$ when it first moves left.
Then the total distance traveled by both drones before they meet is $2(z - w) - (y - x)$,
which means the drones meet at time $z - w - (y - x) / 2 \le z - w \le 1$.
\end{proof}

\begin{lemma}
  \label{lemma:moving:left:after:separation}
Suppose that at time $t$, drone $j$ is moving to the left and
drone $j+1$ is not together with drone $j$.
Suppose also that $j$ and $j+1$ have met by time $t$.
Let $t'$ be the last time before time $t$ that drones $j$ and $j+1$ bounced or separated.
Then $j$ has been moving left since time $t'$.
\end{lemma}

\begin{proof}
If at some point between $t'$ and $t$ drone $j$ was moving to the right,
something must have turned it to the left.
But that can only have been a meet or separation or bounce event.
If it was a meet event,
the fact that $j$ and $j+1$ are not together at time $t$ means there was also
a separation event.
Both situations contradict
the fact that $t'$ is the last time before time $t$ that drones $j$ and $j+1$
bounced or separated.
\end{proof}

\begin{lemma}
  \label{lemma:induction:step}
Suppose $j < n$ and at time $t$, drones $1,\ldots, j$ are left synchronized
and drone $j$ and $j + 1$ have met.
Then at time $t + 1/n$, drones $j + 1$ is left synchronized as well.
\end{lemma}

\begin{proof}
Suppose drone $j$ is left synchronized.
If it is moving to the right, it will be at or to the right of its right endpoint within time $1 / n$,
possibly having met drone $j + 1$ along the way.
At that point drone $j + 1$ is left synchronized,
by Lemmas~\ref{lemma:left:synchronized:after:separation} and \ref{lemma:left:synchronized:sufficient:b}.
If drone $j$ is moving to the left and it is together with drone $j + 1$,
drone $j + 1$ is left synchronized at time $t$ by Lemma~\ref{lemma:left:synchronized:sufficient:a}.

Finally, suppose drone $j$ is moving to the left and is not together with drone $j + 1$.
Since we are assuming drones $j$ and $j + 1$ have met by time $t$,
there is a $t' < t$ where drones $j$ and $j + 1$ bounced or separated last.
By Lemma~\ref{lemma:moving:left:after:separation}, drone $j$ has been moving left since time $t'$.
Since drone $j$ is at or to the right of its left endpoint at time $t$,
it was to the right of its left endpoint between time $t'$ and $t$.
Since drone $j$ is left synchronized at time $t$, this shows that it was already
left synchronized at $t'$.
By Lemma~\ref{lemma:left:synchronized:after:separation},
drone $j + 1$ was also left synchronized at time $t'$, and hence is left synchronized at time $t$.
\end{proof}

Since drone $1$ is always left synchronized and all the drones have met by time 1,
by induction on $i < n$
we have that drones $1, \ldots, i$ are left synchronized at time $1 + (i-1) / n$.
Taking $i = n$ yields Theorem~\ref{theorem:upper:bounds}.

It is not hard to show that Theorem~\ref{theorem:upper:bounds} is sharp. To attain the worst-case behavior, let all $n$ drones start arbitrarily close to the left border, moving right independently. After close to one unit of time they reach the right border, at which point the rightmost drone turns left and quickly meets all the others. The group then moves to the left, with each drone separating from the group at its left endpoint. Drones 1 and 2 separate at their common endpoint at time arbitrarily close to $2 - 1/n$, at which point all the drones are synchronized. This is exactly the worst-case scenario presented by Kingston et al.

We end this section by proving Theorem~\ref{theorem:upper:bounds:combined}.
In that theorem we assume that all drones start with incorrect information.
Without such an assumption, the theorem is false; for example, if all drones start with correct information,
then the fact that Theorem~\ref{theorem:upper:bounds} is sharp means that Theorem~\ref{theorem:upper:bounds:combined} does not hold in that case.
However, the condition is not very strong,
and we can find worst-case configurations for phase 1 where this condition holds,
since we can modify the correct information by adding a small amount
to the estimated left and right endpoints of the interval.

\begin{lemma}
  Suppose all drones start with incorrect information, $j < n$,
  and at time $t$ all drones have correct information and
  drones $1,\ldots, j$ are left synchronized.
  Then at time $t+1-j/n$, all drones are synchronized.
\end{lemma}
\begin{proof}
  The only way that all drones have correct information at time $t$ is that drone 1 hits the left border,
  and then this correct left information propagates through all drones to drone $n$.
  Similarly, the correct right information has propagated through all drones from drone $n$ to drone 1.
  For two consecutive drones $i$ and $i+1$ this means that at some time $t'\le t$ they were together
  with both correct left and correct right information.
  Any time after $t'$ the drones $i$ and $i+1$ had correct information.
  Inspecting the proof of Lemma~\ref{lemma:induction:step}, we see that it also holds in this case
  (the lemmas in this section only use that drones $j$ and $j+1$ have correct information, not that the other drones have correct information).
  Therefore, by induction we see that if drones $1,\ldots, j$ are left synchronized, then at time $t+1-j/n$, all drones are synchronized.
\end{proof}
Since drone 1 is always left synchronized, this implies Theorem~\ref{theorem:upper:bounds:combined}.
In Section~\ref{section:lower:bounds} we present an example where all drones get correct information at time $4-1/n$, but drone 2 is already left synchronized at that time. Theorem~\ref{theorem:upper:bounds:combined} then guarantees that the drones will be synchronized at time $(4-1/n) + (1-2/n) = 5-3/n$.

\section{Lower bounds}
\label{section:lower:bounds}
  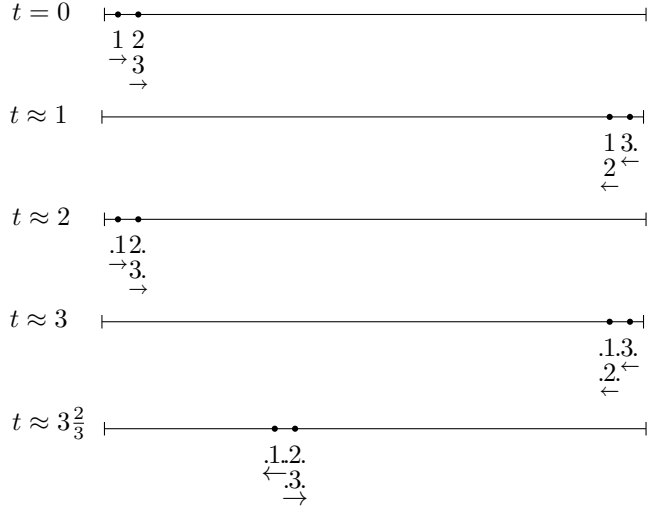
\begin{figure}
    \begin{center}
    \begin{tikzpicture}[scale=0.9]
    \node[anchor=west] at (-1.5,0) {$t=0$};
    \draw (0, 0) -- (8, 0);
    \draw (0,-0.1) -- (0, 0.1);
    \draw (8, -0.1) -- (8, 0.1);
    \fill (0.2, 0) circle (0.3ex) node [below=5pt] {$\Shortunderstack[c]{1 {$\shortrightarrow$}}$};
    \fill (0.5, 0) circle (0.3ex) node [below=5pt] {$\Shortunderstack[c]{2 3 {$\shortrightarrow$}}$};
  \end{tikzpicture}
  \begin{tikzpicture}[scale=0.9]
    \node[anchor=west] at (-1.5,0) {$t\approx 1$};
    \draw (0, 0) -- (8, 0);
    \draw (0,-0.1) -- (0, 0.1);
    \draw (8, -0.1) -- (8, 0.1);
    \fill (7.5, 0) circle (0.3ex) node [below=5pt] {$\Shortunderstack[c]{1 2 {$\shortleftarrow$}}$};
    \fill (7.8, 0) circle (0.3ex) node [below=5pt] {$\Shortunderstack[c]{3\kern-0.07em. {$\shortleftarrow$}}$};
  \end{tikzpicture}
  \begin{tikzpicture}[scale=0.9]
    \node[anchor=west] at (-1.5,0) {$t\approx 2$};
    \draw (0, 0) -- (8, 0);
    \draw (0,-0.1) -- (0, 0.1);
    \draw (8, -0.1) -- (8, 0.1);
    \fill (0.2, 0) circle (0.3ex) node [below=5pt] {$\Shortunderstack[c]{.\kern-0.07em1 {$\shortrightarrow$}}$};
    \fill (0.5, 0) circle (0.3ex) node [below=5pt] {$\Shortunderstack[c]{2\kern-0.07em. 3\kern-0.07em. {$\shortrightarrow$}}$};
  \end{tikzpicture}
  \begin{tikzpicture}[scale=0.9]
    \node[anchor=west] at (-1.5,0) {$t\approx 3$};
    \draw (0, 0) -- (8, 0);
    \draw (0,-0.1) -- (0, 0.1);
    \draw (8, -0.1) -- (8, 0.1);
    \fill (7.5, 0) circle (0.3ex) node [below=5pt] {$\Shortunderstack[c]{.\kern-0.07em1\kern-0.07em. .\kern-0.07em2\kern-0.07em. {$\shortleftarrow$}}$};
    \fill (7.8, 0) circle (0.3ex) node [below=5pt] {$\Shortunderstack[c]{3\kern-0.07em. {$\shortleftarrow$}}$};
  \end{tikzpicture}
  \begin{tikzpicture}[scale=0.9]
    \node[anchor=west] at (-1.5,0) {$t\approx 3\frac23$};
    \draw (0, 0) -- (8, 0);
    \draw (0,-0.1) -- (0, 0.1);
    \draw (8, -0.1) -- (8, 0.1);
    \fill (2.516, 0) circle (0.3ex) node [below=5pt] {$\Shortunderstack[c]{.\kern-0.07em1\kern-0.07em. {$\leftarrow$}}$};
    \fill (2.816, 0) circle (0.3ex) node [below=5pt] {$\Shortunderstack[c]{.\kern-0.07em2\kern-0.07em. .\kern-0.07em3\kern-0.07em. {$\rightarrow$}}$};
  \end{tikzpicture}
  \end{center}
  \caption{A depiction of the worst-case scenario for phase 1 that we found for $n=3$.
    A dot to the right (left) of a drone indicates that it has correct right (left) information.}
  \label{figure:threedrones}
\end{figure}

Recall that Davis, Humphrey, and Kingston \cite{davis:et:al:19} produced counterexamples to
Conjecture~\ref{conjecture:a} by presenting configurations of three drones
that do not all have correct information about the state of affairs
until time $3 + 1/2$. In this section, we improve this
lower bound for three drones to $3 + 2/3$ and modify their example to provide
new lower bounds for every $n \ge 3$. Our strategy is similar to theirs.
We let all the drones start near the left border of the interval,
moving to the right.
We put the drones in groups that separate just before they hit the right border,
so that most drones in each group do not learn the information about the right border
from the group to their right.
After this, the drones will move back in groups to the left border,
and then we try to repeat this process:
let all groups separate just before they hit the left border,
and then the drones move to the right again.

It is conceivable that things can be arranged so that $n$ drones
shuttle back and forth on the order of $n$ times before information
about the right border has propagated all the way to the left and vice-versa.
To rule out such a possibility we have to take into consideration the
specific algebraic calculations described in Section~\ref{section:the:problem},
and the extent to which they constrain the drones' estimates.
In Section~\ref{section:algebraic:calculation}, we consider a
simplified version of the calculation, and show that,
with that version (which does not directly apply to the original problem),
the type of behavior described above cannot occur.
But our counterexamples show that with the actual calculations
presented in Section~\ref{section:the:problem},
we can obtain slightly worse behavior. In particular, it can
take close 4 units of time before the drones have correct estimates.

Let us start with the smallest interesting situation, with $n = 3$ drones.
Our counterexample is summarized in Fig.~\ref{figure:threedrones}.
\begin{itemize}
  \item The drones start near the left border, moving to the right,
  with drones 2 and 3 moving as a group.
  \item Just before the drones hit the right border, drones 2 and 3 separate,
  and drone 2 meets drone 1.
  Drone 3 hits the border, learning the true position of the right border,
  and all drones move to the left.
  \item Just before the drones hit the left border, drones 1 and 2 separate,
  and drone 2 meets drone 3 (learning the true position of the right border).
  Drone 1 hits the border, learning the true position of the left border,
  and all drones move to the right again.
  \item This process repeats at the right border,
  after which drones 1 and 2 have complete knowledge of the left and the right border.
  \item Drones 1 and 2 move to their actual common endpoint, and separate.
  Then drone 3 also learns the true left endpoint, and moves to its common endpoint with drone 2.
\end{itemize}
To achieve this, choose $N$ large, and let $\delta=1/N$.
Let drone $1$ start at $x_1=0$ and drone $2$ start at $x_2=\delta$, both with direction $d_i=1$.
The initial information for drones $1$ and $2$ are as follows:
\begin{align*}
  ((a_1,\ell_1),(b_1,m_1))&=((-2+4\delta,N),(1,N))\\
  ((a_2,\ell_2),(b_2,m_2))&=((0,N),(2-2\delta,N))
\end{align*}
The initial information for drone 3 is fully determined
by the fact that it starts in the same group as drone 2:
\begin{align*}
x_3&=x_2\\
((a_3,\ell_3),(b_3,m_3)) &= ((a_2,\ell_2+1),(b_2,m_2-1)).
\end{align*}
This initial information results in the behavior in Fig.~\ref{figure:threedrones}.
Drone 2 turns around at positions $1-\delta$, then $\delta$, then $1-3\delta$ and then $1/3$.
Only after the last time that drone 2 turns around do all drones have correct estimates.
This means that by taking $N$ large enough,
it can require (arbitrarily close to) $3 + 2/3$ units of time before every drone has correct estimates.
It takes another $1/3$ unit of time until drone 3 is synchronized.
Therefore, the drones need time (arbitrarily close to) $4$ before all drones are synchronized.

\begin{figure}
  \begin{center}
  \includegraphics[width=0.45\textwidth]{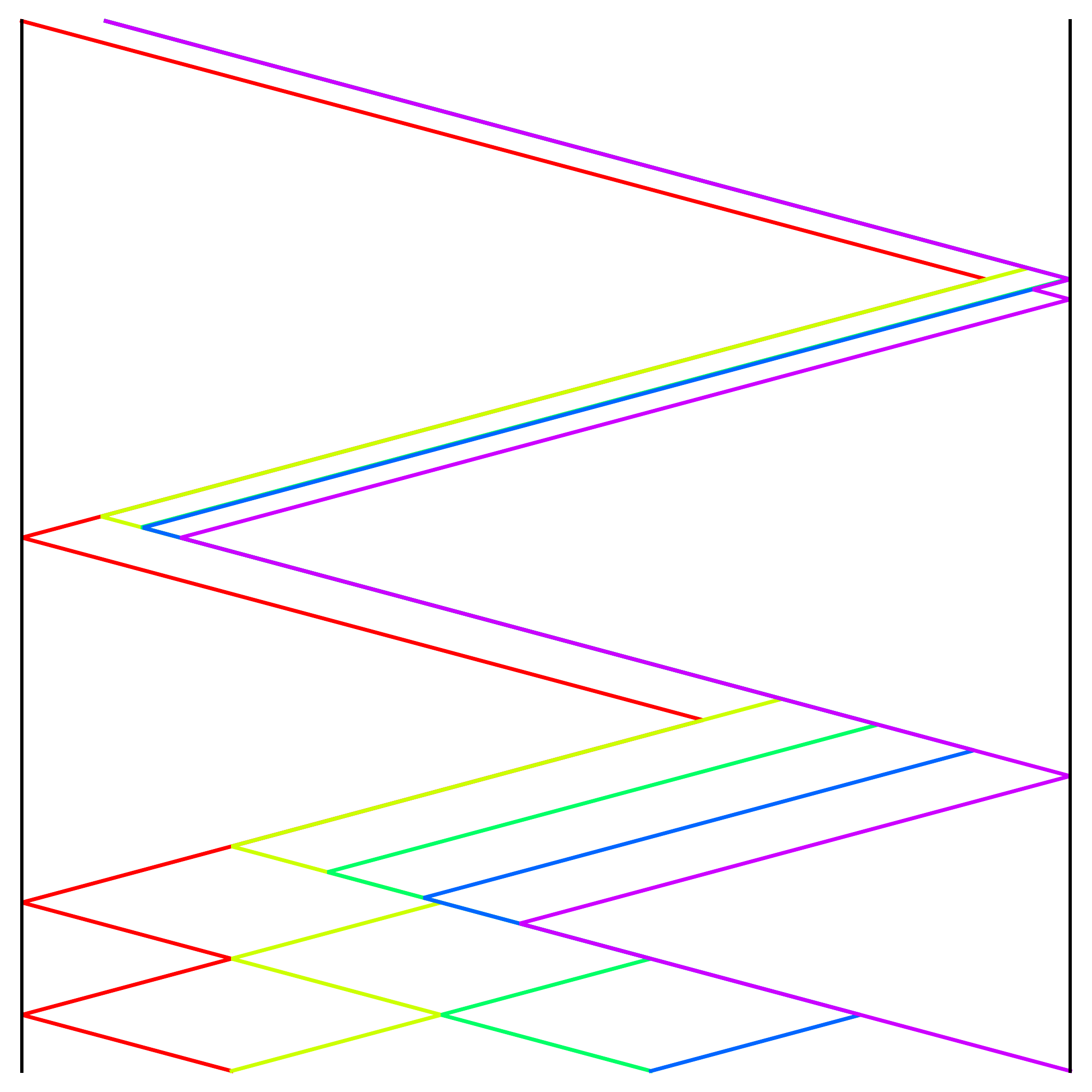}
  \includegraphics[width=0.45\textwidth]{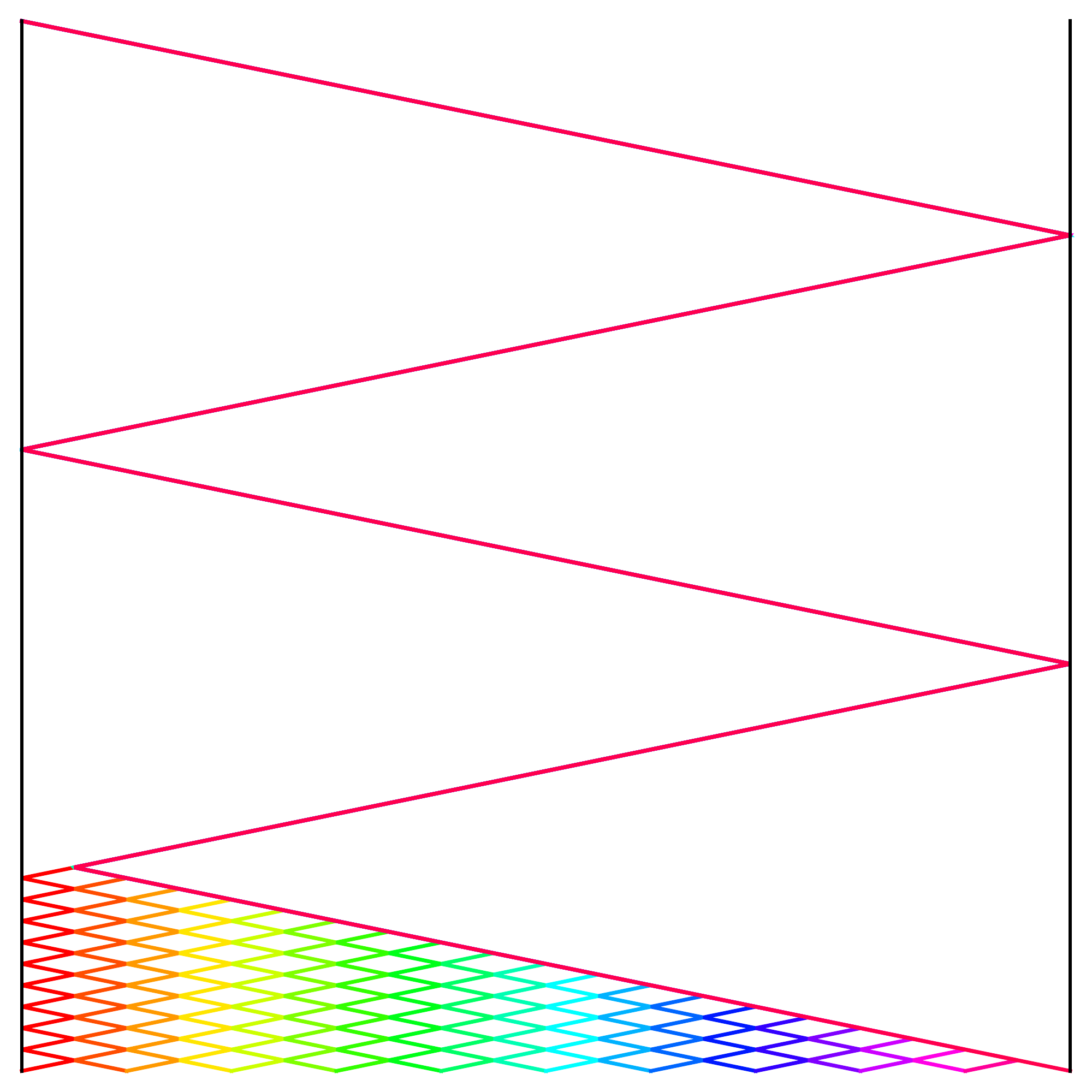}
  \end{center}
  \caption{A depiction of the worst-case scenario we found for $n=5$ (left) and $n=20$ (right).
  In the left diagram $N=25$ and $N=10^6$ in the right diagram.
  In the left diagram $N$ is chosen to be small to show the behavior better,
  but this choice of $N$ gives rise to an artifact that doesn't happen for high $N$, which is that drone $n$ turns around a second time around time 1.}
  \label{figure:moredrones}
\end{figure}

We can generalize this by adding more drones at the same position as drone 2 and 3.
If there are more than 2 drones,
the positions at which the intermediate drones turn are still roughly the same,
except that the fourth turn is at $x=1/n$.
This means that it takes time (arbitrarily close to) $4-1/n$
until all drones have correct estimates.
After that it takes time $1-2/n$ until drone $n$ is synchronized,
for a total of (arbitrarily close to) $5-3/n$ for all drones to be synchronized.

The situation is displayed in Fig.~\ref{figure:moredrones} for $(n,N)=(5,25)$ and $(n,N)=(20,10^6)$.
In these diagrams every drone has a unique color
(though if drones travel close to each other only the rightmost drone is drawn),
time flows down along the vertical axis, and the interval is on the horizontal axis.

We have explored other methods for finding lower bounds with more than 3 drones
using Mathematica~\cite{Mathematica}.
One method we used was to manually pick all starting points and desired points of separation,
enter all constraints into Mathematica,
and use Mathematica's \lstinline{FindInstance} function or \lstinline{Reduce} function
to find a solution.
However, this problem quickly becomes intractable for these functions because of the subtle
interplay between the equations like $R(\alpha_2,\beta_2) = 1-\delta$ and the inequalities like
$a_i\le x_i \le b_i$ and $m_i\ge 0$.
We resorted to heuristics, strategically choosing some values and letting Mathematica solve for
the remaining variables. For example, we often set the estimated number drones to the left or right
to a specific drone to be fixed large integers.
We have found an example with five drones in three groups,
which also takes time roughly time $4-\frac15$ for phase 1, and time $5-\frac35$ for phase 1 + 2,
similar to our previous example.
Drone 1 starts alone at position 0, drones 2 and 3 start together at position $\epsilon=\frac1{100}$
and drone 4 and 5 start together at position $2\epsilon$.
The initial estimates for drones 1, 2 and 4 are
\begin{align*}
(\alpha_1,\beta_1)&=((-250098,10^8),(1,10^6)\\
(\alpha_2,\beta_2)&=((-249999,10^8),(2501,10^6))\\
(\alpha_4,\beta_4)&=((-232446,10^8),(13.94,5569))
\end{align*}
The initial estimates for drones 3 and 5 are determined by those of drones 2 and 4.
The resulting diagram, not shown, is similar to that of Fig.~\ref{figure:moredrones}.

\section{Towards an upper bound on phase 1}
\label{section:algebraic:calculation}

\newcommand{\ipos}[4]{\frac{#1 #4 + #3 #2 }{#2 + #4}}
\newcommand{\iposi}[2]{\ipos{a_{#1}}{\ell_{#1}}{b_{#2}}{m_{#2}}}

The results described in Section~\ref{section:lower:bounds}
involve crafting examples where incorrect estimates on the part of the drones
lead groups of drones to misjudge their common endpoints and shuttle back
and forth across the interval.
From a combinatorial perspective, it is not hard to imagine sequences of events
where $n$ drones keep regrouping in pairs and traveling back and forth up to $n$
times before all the border information has propagated to all the drones.
The question is whether the algebraic calculations of the common endpoints
make it possible to realize this behavior. Getting a bound on phase 1
that is independent of $n$ requires ruling this out.

In this section, we take small steps towards obtaining a better
understanding of the algebraic constraints.
If $\alpha = (a, \ell)$ is a pair consisting of a real number and a
nonnegative integer,
we will write $\alpha^+$ for $(a, \ell+1)$ and $\alpha^-$ for $(a, \ell-1)$.
Remember that a pair $(\alpha, \beta)$ represents a drone's estimates
as to the left border,
the number of drones to the left,
the right border, and the number of drones to the right.
Remember also that we write
$L(\alpha, \beta)$ for the left endpoint of the drone's interval based on
that estimate, and $R(\alpha, \beta)$ for the right endpoint.
We will also write $0$ for the estimate $(0, 0)$ adopted by the leftmost drone
when it reaches the left border, and $1$ for the estimate $(1, 0)$
adopted by the rightmost drone when it reaches the right border.

Consider the example in Fig.~\ref{figure:threedrones}.
The estimates of the three drones on each line can be represented as follows:
\begin{align*}
  &(\alpha_1, \beta_1)   & &(\alpha_2  , \beta_2) & &(\alpha_2^+   , \beta_2^-) \\
  &(\alpha_1, \beta_2^+) & &(\alpha_1^+, \beta_2) & &(\alpha_2^+   , 1) \\
  &(0       , \beta_2^+) & &(\alpha_1^+, 1^+)     & &(\alpha_1^{++}, 1) \\
  &(0       , 1^{++})    & &(0^+       , 1^+)     & &(\alpha_1^{++}, 1) \\
  &(0       , 1^{++})    & &(0^+       , 1^+)     & &(0^{++}       , 1)
\end{align*}
What makes the example effective is that:
\begin{itemize}
  \item $R(\alpha_2, \beta_2) = L(\alpha_2^+, \beta_2^-)$ is close to 1.
  \item $R(\alpha_1, \beta_2^+) = L(\alpha_1^+, \beta_2)$ is close to 0.
  \item $R(\alpha_1^+, 1^+) = L(\alpha_1^{++}, 1)$ is close to 1.
  \item $R(0, 1^{++}) = L(0^+, 1^+) = 1/3$.
\end{itemize}
To understand the extent to which we can or cannot improve the lower bound,
we need to understand the constraints that arise when drones share
information in such a way.

The $+1$ and $-1$ terms make calculation more difficult,
so we focus on a simpler approximation to the problem.
Given $\alpha = (a, \ell)$ and $\beta = (b, m)$, write
\[
P(\alpha, \beta) = \ipos{a}{\ell}{b}{m}
\]
for an approximation to $L(\alpha, \beta)$ and $R(\alpha, \beta)$.
In other words, we ignore the terms $+1$ in the calculations
in Section~\ref{section:the:problem},
which is reasonable when $\ell$ and $m$ are large.

Now, suppose three groups of drones start with estimates that are roughly
$(\alpha_1, \beta_1)$, $(\alpha_2, \beta_2)$, and $(\alpha_3, \beta_3)$.
Sharing information between the first two yields roughly $(\alpha_1, \beta_2)$,
and sharing information between the second two yields roughly $(\alpha_2, \beta_3)$.
Sharing information between these again yields roughly $(\alpha_1, \beta_3)$.
Under these simplifications, we can ask the following question:
are there choices of $(\alpha_1, \beta_1)$, $(\alpha_2, \beta_2)$ and
$(\alpha_3, \beta_3)$ such that
\begin{itemize}
\item $P(\alpha_1, \beta_1)$, $P(\alpha_2, \beta_2)$, and
  $P(\alpha_3, \beta_3)$ are close to 1,
\item $P(\alpha_1, \beta_2)$ and $P(\alpha_2, \beta_3)$ are close to 0, and
\item $P(\alpha_1, \beta_3)$ is close to 1?
\end{itemize}
The following theorem shows that the answer is negative.

\begin{theorem}
\label{theorem:algebraic:simplification}
If
\[
  \max(P(\alpha_1, \beta_2), P(\alpha_2, \beta_3)) \le \min(P(\alpha_2, \beta_2), P(\alpha_1, \beta_3))
\]
then
\[
  P(\alpha_1, \beta_2) = P(\alpha_2, \beta_3) = P(\alpha_2, \beta_2) = P(\alpha_1, \beta_3)
\]
\end{theorem}
\noindent In this theorem only the four values $P(\alpha_2, \beta_2)$, $P(\alpha_1, \beta_2)$,
$P(\alpha_2, \beta_3)$, and $P(\alpha_1, \beta_3)$ are mentioned,
and the constraints are weaker than stated in the original question.
We provide three proofs of this theorem.

\begin{proof}[First proof]
The statement $P(\alpha, \beta) = c$ is equivalent to
$(c-a) / \ell = (b-c) / m$.
If we ignore the length of the drone's own interval, $(c - a) / \ell$ is
the length of the interval for each drone to the left of $a$,
and $(b - c) / m$ is the length of the interval for each drone to the
right of $b$. So the statement $P(\alpha, \beta) = c$ says that
$c$ has the property that these two lengths are the same.

Let $v(\alpha,c) = (c-a) / \ell$ and $w(\beta,c)= (b-c) / m$,
so $v(\alpha,c)$ is strictly increasing in $c$ and
$w(\beta,c)$ is strictly decreasing in $c$.
Write
\begin{align*}
P(\alpha_1, \beta_3)&=r_1, &P(\alpha_2, \beta_2)&=r_2,\\
P(\alpha_1, \beta_2)&=\ell_1, &P(\alpha_2, \beta_3)&=\ell_2.
\end{align*}
The hypothesis of the lemma is that $\ell_i \le r_j$ for each $i$ and $j$.
We have the following chain of inequalities:
\begin{align*}
  v(\alpha_2,r_2)&=w(\beta_2,r_2)\le w(\beta_2,\ell_1)\\
  &=v(\alpha_1,\ell_1)\le v(\alpha_1,r_1)\\
  &=w(\beta_3,r_1)\le w(\beta_3,\ell_2)\\
  &=v(\alpha_2,\ell_2)\le v(\alpha_2,r_2).
\end{align*}
Therefore, all inequalities must be equalities, and the fact that
$v$ and $w$ are strictly monotone implies $\ell_2 = r_2 = \ell_1 = r_1$.
\end{proof}

\begin{proof}[Second proof]
Write
\[
I'(\alpha, \beta) = \frac{b - a}{\ell + m}
\]
for the length of each drone's interval based on the information $\alpha, \beta$, and notice that
\begin{equation}
  \label{equation:ipos:expand}
  \tag{$\star$}
  P(\alpha, \beta) = a + \ell I'(\alpha, \beta) = b - m I'(\alpha, \beta).
\end{equation}
This enables us to do comparisons when either $\alpha$ or $\beta$ does not change:
\begin{itemize}
  \item From $P(\alpha, \beta) \le P(\alpha, \beta')$ we can infer $I'(\alpha, \beta) \le I'(\alpha, \beta')$.
  \item From $P(\alpha, \beta) \le P(\alpha', \beta)$ we can infer $I'(\alpha', \beta) \le I'(\alpha, \beta)$.
\end{itemize}
So from
\begin{align*}
  P(\alpha_1, \beta_2) &\le P(\alpha_1, \beta_3), &P(\alpha_2, \beta_3) &\le P(\alpha_1, \beta_3),\\
  P(\alpha_2, \beta_3) &\le P(\alpha_2, \beta_2), &P(\alpha_1, \beta_2) &\le P(\alpha_2, \beta_2)
\end{align*}
we can conclude
\begin{align*}
  I'(\alpha_1, \beta_2) &\le I'(\alpha_1, \beta_3) \le I'(\alpha_2, \beta_3) \le I'(\alpha_2, \beta_2) \\
  &\le I'(\alpha_1, \beta_2)
\end{align*}
and so all these quantities are the same. Using (\ref{equation:ipos:expand}), this means that any two expressions $P(\alpha_i, \beta_j)$ that have either $\alpha_i$ or $\beta_j$ in common are equal, so we have $P(\alpha_1, \beta_2) = P(\alpha_1, \beta_3) = P(\alpha_2, \beta_3) = P(\alpha_2, \beta_2)$, as required.
\end{proof}

We are grateful for Reid Barton for providing us with the following physical interpretation.

\begin{proof}[Third proof]
  If we think of $\alpha = (a, \ell)$ as representing a mass of $1/\ell$ at position $a$
  on the real number line
  and $\beta = (b, m)$ as representing a mass of $1/m$ at position $b$,
  then $P(\alpha, \beta)$ is their combined center of mass.
  The hypothesis of the lemma says that the
  center of mass of $\alpha_1$ and $\beta_2$ and
  the center of mass of $\alpha_2$ and $\beta_3$
  are both less than the center of mass of $\alpha_2$ and $\beta_2$
  and the center of mass of $\alpha_1$ and $\beta_3$.

  But the combined the center of mass of the first two pairs and
  the center of mass of the second two pairs are both equal to the
  center of mass of $\alpha_1$, $\alpha_2$, $\beta_2$, and $\beta_3$
  all together. Since two objects combined have a center of mass
  between the center of mass of each, all four quantities have to be equal.
\end{proof}

Unfortunately, the omissions of the $+1$s in the definition of $P(\alpha, \beta)$
are an oversimplification. When we restore them to the model, the example at the end of
Section~\ref{section:lower:bounds} is a counterexample
to statements analogous to Theorem~\ref{theorem:lower:bounds}.
So we still need to both refine our calculations and
figure out how to use them to obtain an upper bound to phase 1.

\section{Conclusions}

These problems are harder than we thought it would be. Even bounding phase 2 behavior
was difficult,
since it was hard to control the combinatorial explosion of possibilities
induced by the discrete decision points in the algorithm.
A natural strategy is to look for a complexity measure and show that it is
decreasing over time, but we were unable to find such a measure.
Another idea is to argue that configurations can, without loss of generality (for example, without decreasing the time to synchronization), be reduced to
simpler ones,
possibly in a larger space of states and actions.
We also tried using induction on the number of drones,
or focusing on some salient feature of the event history.
We could not get any of these strategies to work.
It was surprising to us that the argument presented in Section~\ref{section:upper:bounds}
works,
because it involves ignoring everything about the first unit of time other than
the fact that each pair of drones has met.
We stumbled upon it only by finding a more complicated argument and simplifying it to the shorter one,
and then we determined that the longer argument was incorrect.

Bounding phase 1 promises to be even harder, since it requires taking
the algebraic constraints into consideration as well as the combinatorial behavior.
Despite the difficulty of the problem,
we are hopeful that the method used to bound the phase 2 behavior,
together with a better understanding of the algebraic constraints,
will lead to a solution.




\paragraph{Acknowledgments.} This work was supported in part by the AFOSR under grant FA9550-18-1-0120
and the Sloan Foundation under grant G-2018-10067.

\bibliographystyle{plain}
\bibliography{progress}

\end{document}